\documentclass{llncs}

\usepackage{amsmath}
 \usepackage{amssymb}
\usepackage{microtype}%if unwanted, comment out or use option "draft"
\usepackage{graphicx}

\newtheorem{observation}{Observation} 
\newtheorem{clm}{Claim} 

\newcommand{\evc}[0]{\operatorname{evc}}
\newcommand{\mvc}[0]{\operatorname{mvc}}
\newcommand{\cvc}[0]{\operatorname{cvc}}
\newcommand\Osq{\mathbin{\text{\scalebox{.84}{$\square$}}}}

\title{On Graphs whose Eternal Vertex Cover Number and Vertex Cover Number Coincide}
\titlerunning{On Graphs whose Eternal Vertex Cover Number and Vertex Cover Number Coincide} %optional, in case that the title is too long; the running title should fit into the top page column

\author{Jasine Babu\inst{1} \and L. Sunil Chandran \inst{2} \and Mathew Francis \inst{3} \and Veena Prabhakaran \inst{1} \and Deepak Rajendraprasad \inst{1} \and J. Nandini Warrier \inst{4}}
\institute{Indian Institute of Technology Palakkad, India \email{jasine@iitpkd.ac.in,111704003@smail.iitpkd.ac.in,deepak@iitpkd.ac.in}
\and Indian Institute of Science, Bangalore, India \email{sunil@csa.iisc.ernet.in}
\and Indian Statistical Institute, Chennai, India \email{mathew@isichennai.res.in}
\and National Institute of Technology Calicut, India \email{nandini.wj@gmail.com}}

\authorrunning{J.\,Babu, L.\,S.\.Chandran, M.\.Francis, V.\.Prabhakaran, D.\.Rajendraprasad  and J.\.N.\.Warrier} %mandatory. First: Use abbreviated first/middle names. Second (only in severe cases): Use first author plus 'et. al.'

\begin{document}
\maketitle
\begin{abstract}
The eternal vertex cover problem is a variant of the classical vertex cover problem where 
a set of guards on the vertices have to be dynamically reconfigured from one vertex cover to another in every round of an attacker-defender game. 
The minimum number of guards required to protect a graph $G$ from 
an infinite sequence of attacks is the eternal vertex cover number of $G$, denoted by $\evc(G)$. 
It is known that, given a graph $G$ and an integer $k$, checking whether $\evc(G) \le k$ is NP-hard. 
However, it is unknown whether this problem is in NP or not. 
Precise value of eternal vertex cover number is known only for certain very basic graph classes like trees, cycles and grids. 

For any graph $G$, it is known that $\mvc(G) \le \evc(G) \le 2 \mvc(G)$, 
where $\mvc(G)$ is the minimum vertex cover number of $G$.
Though a characterization is known for graphs for which  $\evc(G) = 2\mvc(G)$, a characterization of graphs 
for which $\evc(G) = \mvc(G)$ remained open. Here, we achieve such a characterization for a class of graphs that includes chordal graphs and 
internally triangulated planar graphs. For some graph classes including biconnected chordal graphs, 
our characterization leads to a polynomial time algorithm to precisely determine $\evc(G)$ and to determine a safe strategy of guard movement 
in each round of the game with $\evc(G)$ guards. 

The characterization also leads to NP-completeness results for the eternal vertex cover problem for some graph classes including
biconnected internally triangulated planar graphs. To the best of our knowledge, these are the first NP-completeness results known for the 
problem for any graph class. 
\footnote{An initial version containing some results included in this paper appeared in CALDAM 2019 \cite{BabuCFPRW19}.}
\end{abstract}
% mandatory: Please provide 1-5 keywords
\section{Introduction}\label{sec:intro}
A vertex cover of a graph $G(V, E)$ is a subset $S \subseteq V$ such that for every edge in $E$, at least one of its endpoints is in $S$.
A minimum vertex cover of $G$ is a vertex cover of $G$ of minimum cardinality and its cardinality is the minimum vertex cover number of $G$, denoted by $\mvc(G)$.
Equivalently, if we imagine that a guard placed on a vertex $v$ can monitor all edges incident at $v$,
then $\mvc(G)$ is the minimum number of guards required to ensure that all edges of $G$ are monitored.

The eternal vertex cover problem is an extension of the above formulation in the context of a multi-round game, where mobile guards placed on a subset
of vertices of $G$ are trying to protect the edges of $G$ from an attacker.
This problem was first introduced by Klostermeyer and Mynhardt \cite{Klostermeyer2009}. We focus on a well studied variant of the game 
in which more than one guard occupying a vertex simultaneously is disallowed at any point of time. 
Guards are initially placed by the defender on some vertices, with at most one guard per vertex. 
The total number of guards on vertices remain the same throughout the game. 
In each round of the game, the attacker chooses an edge to attack. In response, the defender has to move at least one guard across the attacked edge. 
Other guards can either remain in their current locations or move to an adjacent vertex.
The movement of all guards in a round is assumed to happen in parallel. 
Then the game proceeds to the next round of attack-defense. The defender wins if any sequence of attacks can be defended.
If an attack cannot be defended in some round, the attacker wins.

If $\mathcal{C}$ is a family of vertex covers of $G$ of the same cardinality, such that the defender can choose any vertex cover from $\mathcal{C}$ as the starting configuration
and successfully keep on defending attacks forever by moving among configurations in $\mathcal{C}$ itself, then $\mathcal{C}$ is an \textit{eternal vertex cover class} 
of $G$ and each vertex cover in $\mathcal{C}$ is an \textit{eternal vertex cover} of $G$. If $S$ is an eternal vertex cover belonging to an eternal vertex cover class $\mathcal{C}$, 
we say that $S$ is a configuration in $\mathcal{C}$. Eternal vertex cover number of $G$ is the minimum cardinality of an eternal vertex cover of $G$.

Clearly, if the vertices occupied by the guards do not form a vertex cover at the beginning of each round, 
there is an attack which cannot be defended, namely an attack on an edge that has no guards on its end points. 
Therefore, it is easy to see that for any graph $G$, $\mvc(G) \le \evc(G)$.

% The above informal description of a legal transition between two configurations can be made formal as follows. 
% Let $S_i$ and $S_j$ be vertex covers of $G$ of the same cardinality with $k=|S_i \setminus S_j|$ and $uv$ be an edge of $G$.
% The transition from $S_i$ to $S_j$ is legal for defending an attack on $uv$ if and only if  
% (i) there are $k$ vertex disjoint paths from 
% $S_i \setminus S_j$ to $S_j \setminus S_i$ in which all internal vertices are from $S_i \cap S_j$ and 
% (ii) if $\{u, v\} \nsubseteq S_i$ or $\{u, v\} \nsubseteq S_j$ then one of the $k$ paths contains the edge $uv$. 

Klostermeyer and Mynhardt \cite{Klostermeyer2009} showed that, for $C_n$, a cycle on $n$ vertices with $n\ge3$, $\evc(C_n)=\mvc(C_n)=\lceil{\frac{n}{2}}\rceil$ and for 
any tree on $n$ vertices with $n\ge2$, eternal vertex cover number is one more than its number of internal vertices. In particular, for a path on an odd number of vertices,
its eternal vertex cover number is twice its vertex cover number. 
They also showed that, for any graph $G$, $\mvc(G) \le \evc(G) \le 2 \mvc(G)$. 
From the examples of cycles and paths, it is clear that even for bipartite graphs, both the lower bound and upper bound mentioned above are tight.

Fomin et al.~\cite{Fomin2010} discusses the computational complexity and derives some algorithmic results for the eternal 
vertex cover problem. They use a variant of the eternal vertex cover problem in which more than one guard can be placed 
on a single vertex. They showed that given a graph $G(V,E)$ and an integer $k$, it 
is NP-hard to decide whether $\evc(G) \le k$. 
The paper gave an exact algorithm with $2^{O(n)}$ time complexity and exponential space complexity and
gave an FPT algorithm to solve the eternal 
vertex cover problem,  with eternal vertex cover number as the parameter. 
They also describe a simple polynomial time 2-factor approximation algorithm for the eternal vertex cover problem, using maximum matchings. 
The above results can be carried forward (with minor modifications in proofs) to the original model 
which allows at most one guard per vertex. It is not yet known whether the decision problem is in NP, though it is known that the 
problem is in PSPACE \cite{Fomin2010}. 
It is also unknown whether the eternal vertex cover problem for bipartite graphs is NP-hard. 
Some related graph parameters based on multi-round attacker-defender games and their relationship with eternal vertex cover 
number were investigated by Anderson et al.~\cite {Anderson2014} and 
Klostermeyer et al.~\cite{Klostermeyer2011}.

Klostermeyer and Mynhardt \cite{Klostermeyer2009} gave a characterization for graphs $G$ which have $\evc(G)=2\mvc(G)$. 
The characterization follows a nontrivial constructive method starting from any tree $T$ which requires $2\mvc(T)$ guards to protect it.
They also give a few examples of graphs $G$ for which $\evc(G)=\mvc(G)$ such as complete graph on $n$ 
vertices ($K_n$), Petersen graph, $K_m \Osq K_n$, $C_m \Osq C_n$ (where $\Osq$ represents the box product) and $n \times m$ grid, if $n$ or $m$ is even. 
However, they mention that an elegant characterization of graphs for which $\evc(G)=\mvc(G)$ seems to be difficult.  

Here, we achieve such a characterization that works for a class of graphs $\mathcal{F}$, that includes locally connected graphs, chordal graphs and
internally triangulated planar graphs. Without loss of generality, we only consider
connected graphs for this characterization. 
The graph class $\mathcal{F}$ consists precisely of all graphs $G$ for which each minimum vertex cover of $G$ that contains all its cut vertices
induces a connected subgraph in $G$. 
%in every block of the graph, each minimum vertex cover of the block induces a connected subraph. 
The characterization is simple to state (see Theorem~\ref{thm:evc-cut-vertices-char}) and can be used to show that given a connected graph $G$ 
for which every minimum vertex cover is connected, deciding whether $\evc(G)\le k$ is in NP. 
Further, if $\mathcal{H}$ is a hereditary subclass of $\mathcal{F}$ for which
minimum vertex cover computation can be done in polynomial time, then using our characterization, it can be shown that given a graph $G$ in $\mathcal{H}$,
deciding whether $\evc(G)=\mvc(G)$ can be done in polynomial time. 
It can also be shown that given a biconnected graph $G$ in $\mathcal{H}$, $\evc(G)$ can be computed in polynomial time.

In particular, our characterization has the following implications:
\begin{itemize}
\item For chordal graphs, deciding whether $\evc(G)=\mvc(G)$ can be done in polynomial time. If the parameters are equal, then a safe strategy of guard movement 
in each round of the game, with $\evc(G)$ guards, can be determined in polynomial time. 
\item For biconnected chordal graphs, $\evc(G)$ can be computed in polynomial time. Further, a safe strategy of guard movement 
in each round of the game, with $\evc(G)$ guards, can be determined in polynomial time. 
\item For internally triangulated planar graphs, deciding whether $\evc(G)=\mvc(G)$ is in $P^{NP}$.
\item Deciding whether $\evc(G) \le k$ is in NP for locally connected graphs, a graph class that includes the class of biconnected internally triangulated 
planar graphs. (A graph is locally connected, if the open neighborhood of each vertex induces a connected subgraph.)
\end{itemize}
Other results included in this paper are the following:
\begin{itemize}
 \item Deciding whether $\evc(G) \le k$ is NP-complete for locally connected graphs and biconnected internally triangulated planar graphs. 
 To the best of our knowledge, these are the first NP-completeness results known for the 
problem for any graph class. Various NP-hardness and approximation hardness results obtained 
 are summarized in Figure~\ref{fig:complexity}.
 \item Klostermeyer and Mynhardt \cite{Klostermeyer2009} had posed a question whether it is necessary for every edge $e$ of $G$ to be present in some maximum matching, to satisfy $\evc(G)=\mvc(G)$. 
We present an example which answers this question in negative. 
\end{itemize}
\section{A necessary condition for $\evc(G) = \mvc(G)$}
%For the rest of the paper, without loss of generality, we assume that the input graph $G(V, E)$ is connected and has at least two vertices. 
In this section, we derive some necessary conditions for a graph $G$ to have $\evc(G) = \mvc(G)$. 
The following is an easy observation. 
\begin{observation}\label{obs:obvious}
Let $G(V, E)$ be a connected graph with at least two vertices. If $\evc(G)=\mvc(G)$, then for every vertex $v \in V$, $G$ has a minimum vertex cover $S_v$ containing $v$.
\end{observation}
\begin{proof}
 Suppose $\evc(G)=\mvc(G)$ and $\mathcal{C}$ is an eternal vertex cover class of $G$ in which each configuration is a vertex cover with exactly $\mvc(G)$ vertices.
 Consider any vertex $v \in V$. If a configuration in $\mathcal{C}$ has no guard on $v$, then following an attack on an edge adjacent to $v$, the next configuration
 should have a guard on $v$, to successfully defend the attack. Since vertex cover in every configuration of $\mathcal{C}$ is a minimum vertex cover, the observation follows. 
\hfill\qed
 \end{proof}
It is easy to see that the simple necessary condition stated above is not sufficient for many graphs. 
For a path $P_n$ on $n$ vertices, where $n>2$ is an even number, each vertex belong to some minimum vertex cover; but still 
$\frac{n}{2}=\mvc(G)<\evc(G)=n-1$. In fact, among graphs which are not biconnected, it is easy to find several such examples.
Therefore, we generalize Observation~\ref{obs:obvious} to get a stronger necessary condition for a graph $G$ with cut vertices.\\

\noindent We first introduce some notations. 
\begin{definition}
 For any subset of vertices $U$ of a graph $G$, we define $\evc_{_{U}}(G)$ 
   as the minimum integer $k$ such that $G$ has an eternal vertex cover class
 $\mathcal{C}$ in which every configuration is a vertex cover of cardinality $k$ that contains all vertices in $U$. 
 We define $\mvc_{_{U}}(G)$ as the minimum cardinality of a vertex cover of $G$ that contains all vertices of $U$.
\end{definition}
Note that when $U=\emptyset$, $\mvc_{_{U}}(G)=\mvc(G)$ and $\evc_{_{U}}(G)=\evc(G)$. 
The following is an easy generalization of Observation~\ref{obs:obvious}.
\begin{observation}\label{obs:evc-subset}
 Let $G(V, E)$ be a connected graph with at least two vertices and $U \subseteq V$. If $\evc_{_{U}}(G) = \mvc_{_{U}}(G)$, then
for every vertex $v \in V \setminus U$, $\mvc_{_{U \cup \{v\}}}(G) = \mvc_{_{U}}(G)$. 
\end{observation}
It is straightforward to obtain a proof of the above observation, by generalizing the proof of Observation~\ref{obs:obvious}.

Next lemma shows that for a graph $G$ for which vertex cover number and eternal vertex cover number coincide, these parameters
also coincide with $\evc_{_{X}}(G)$, where $X$ is the set of cut vertices of $G$. 
\begin{lemma}\label{lem:evc-cut-vertices}
  Let $G(V, E)$ be any connected graph. Let $X \subseteq V$ be the set of cut vertices of $G$. 
 If $\evc(G) = \mvc(G)$, then for any minimum eternal vertex cover class $\mathcal{C}$ of $G$, each 
 configuration of  $\mathcal{C}$ is a vertex cover containing all vertices of $X$. Consequently, $\evc_{_{X}}(G)=\mvc_{_{X}}(G)=\evc(G)=\mvc(G)$.  
\end{lemma}
\begin{proof}
Suppose $\evc(G)=\mvc(G)=k$.
If $X=\emptyset$, the result holds trivially. If $X \neq \emptyset$, we will show that in any minimum eternal vertex cover class $\mathcal{C}$ 
of $G$, all cut vertices of $G$ have to be occupied with guards in all configurations.

Let $x$ be any cut vertex of $G$. Let $H$ be a connected component of $G \setminus x$, 
$H_1=G[V(H) \cup \{x\}]$ and $H_2=G[V \setminus V(H)]$. Note that $H_1$ and $H_2$ are edge-disjoint subgraphs of $G$ with $x$ 
being their only common vertex. Let $k_1=\mvc(H_1)$ and $k_2=\mvc(H_2)$. 
It is easy to see that $k=\mvc(G) \in \{k_1+k_2-1, k_1+k_2\}$. Since $\evc(G) = \mvc(G)$, there must be a vertex cover configuration $S$ in the eternal 
vertex cover class $\mathcal{C}$ such that $x\in S$. Either $|S \cap V(H_1)|=k_1$ or  $|S \cap V(H_2)|=k_2$ or both.
If both $|S \cap V(H_1)|=k_1$ and $|S \cap V(H_2)|=k_2$, then $k=k_1+k_2-1$ and $G$ has no minimum vertex covers without $x$.
This would immediately imply that in every configuration of $\mathcal{C}$, $x$ is occupied by a guard. 

Therefore, without loss of generality, we need to consider the only case when $H_1$ has no minimum vertex cover containing $x$. 
If $x$ is not occupied by a guard in a configuration $S' \in \mathcal{C}$, we must have $|S' \cap V(H_1)|=k_1$ and  $|S' \cap V(H_2)|=k_2$. 
In this configuration, consider an attack on an edge $ux$ in $H_1$. A guard must move to $x$ from $u$. 
This is impossible because $H_1$ had only $k_1$ guards it has no vertex cover containing $x$ of size $k_1$. 
Hence, in this case also, $x$ is occupied by a guard in every configuration of $\mathcal{C}$.

Since $x$ was an arbitrary chosen cut vertex, this implies that all vertices of $X$ must be occupied in all configurations of the
eternal vertex cover class $\mathcal{C}$ and hence $\evc_{_{X}}(G)=\mvc_{_{X}}(G)=\evc(G)=\mvc(G)$.
\hfill\qed\end{proof}
By combining Observation~\ref{obs:evc-subset} and Lemma~\ref{lem:evc-cut-vertices}, we can derive the 
following stronger necessary condition for a graph to have its vertex cover number and eternal vertex cover number coincide.
\begin{theorem}[Necessary Condition]\label{thm:necessary-condn}
   Let $G(V, E)$ be any connected graph with at least two vertices. Let $X \subseteq V$ be the set of cut vertices of $G$. 
 If $\evc(G) = \mvc(G)$, then for every vertex $v \in V \setminus X$, there is a minimum vertex cover $S_v$ of $G$ such that $X \cup \{v\} \subseteq S_v$. 
\end{theorem}
\begin{proof}
Suppose $\evc(G) = \mvc(G)$. Then, by Lemma~\ref{lem:evc-cut-vertices}, we have $\evc_{_{X}}(G)=\mvc_{_{X}}(G)$. 
Hence, by Observation~\ref{obs:evc-subset}, for every vertex $v \in V \setminus X$, $\mvc_{_{X \cup \{v\}}}(G) = \mvc_{_{X}}(G)$. 
Since we also have $\mvc_{_{X}}(G)=\mvc(G)$ by Lemma~\ref{lem:evc-cut-vertices}, this implies that 
for every vertex $v \in V \setminus X$, $\mvc_{_{X \cup \{v\}}}(G) =\mvc(G)$.
 \hfill\qed
\end{proof}
\begin{remark}
 From Theorem~\ref{thm:necessary-condn}, it is evident that if $\evc(G) = \mvc(G)$, 
 then $G$ must have a minimum vertex cover containing all its cut vertices.
\end{remark}
The following is an interesting corollary of Lemma~\ref{lem:evc-cut-vertices} and Observation~\ref{obs:evc-subset}.
\begin{corollary}
 For any connected graph $G$ with at least three vertices and minimum degree one, $\evc(G) \ne \mvc(G)$.
\end{corollary}
The corollary holds because a degree one vertex and its neighbor (which is a cut vertex, if the graph itself is not just an edge) cannot be 
simultaneously present in a minimum vertex cover of $G$. 
\section{Sufficiency of the necessary condition for graph class $\mathcal{F}$}\label{sec:characterization}
In this section, we show that the necessary condition mentioned in Theorem~\ref{thm:necessary-condn} is also sufficient to have $\evc(G) = \mvc(G)$ for 
the graph class $\mathcal{F}$ defined here. We will also discuss some implications of the result when it is
applied to some subclasses of this graph class, like locally connected graphs, chordal graphs,
internally triangulated planar graphs and for graphs for which each minimum vertex cover induces a connected subgraph.
\begin{definition}[Graph class $\mathcal{F}$]
 The graph class $\mathcal{F}$ consists of all connected graphs $G$ for which each minimum vertex cover of $G$ that contains all cut vertices
of $G$ induces a connected subgraph in $G$. 
\end{definition}
For any subset $U \subseteq V$, let $G[U]$ denote the induced subgraph of $G$ on the vertex set $U$.
A vertex cover $S$ of a graph $G$ is called a connected vertex cover if $G[S]$ is connected. 
The \textit{connected vertex cover number} of $G$, $\cvc(G)$, is the size of a minimum cardinality connected vertex cover of $G$. 

The following lemma gives a sufficient condition under which the converse of Observation~\ref{obs:evc-subset} holds. 
The proof of this lemma involves repeated applications of Hall's matching theorem \cite{diestel}. 
\begin{lemma}\label{lem:evc-subset-converse}
 Let $G(V, E)$ be a connected graph with at least two vertices. Let $U \subseteq V$ and
 suppose that every vertex cover $S$ of $G$ of cardinality $\mvc_{_U}(G)$ that contains
 $U$ is connected. If for every vertex $v \in V \setminus U$, $\mvc_{_{U \cup \{v\}}}(G)=\mvc_{_U}(G)$, then $\evc_{_{U}}(G)=\mvc_{_{U}}(G)$. 
\end{lemma}
\begin{proof}
  Let $k=\mvc_{_U}(G)$. Suppose every vertex cover $S$ of $G$ with $U \subseteq S$ and $|S|=k$ is connected. 
  
  Assume that for every vertex $v \in V \setminus U$, $\mvc_{_{U \cup \{v\}}}(G)=k$.
  We will show the existence of an eternal vertex cover class $\mathcal{C}$ of $G$ with exactly $k$ guards such that 
  in every configuration of $\mathcal{C}$, all vertices in $U$ are occupied. 
   
  We may take any vertex cover $S$ of $G$ with $U \subseteq S$ and $|S|=k$ as the starting configuration. 
  It is enough to show that from any vertex cover $S_i$ of $G$ with $U \subseteq S_i$ and $|S_i|=k$, following an attack on an edge $uv$ 
  such that $v \notin S_i \ni u$, we can safely defend the attack by moving to a vertex cover $S_j$ 
  such that $(U \cup \{v\}) \subseteq S_j$ and $|S_j|=k$.
  
   Consider an attack on the edge $uv$ such that $u \in S_i$ and $v \notin S_i$.\\Let $\Gamma =\{S' : S'$ is a vertex cover of $G$ with $|S'|=k$ and $(U \cup \{v\}) \subseteq S'\}$.
     We will show that it is possible to safely defend the attack on $uv$ by moving from $S_i$ to $S_j$, 
     where $S_j \in \Gamma$ is an arbitrary minimum vertex cover such that the cardinality of its symmetric 
     difference with $S_i$ is minimized. 
     
     Let $T=S_i \cap S_j$, $S_i= T \uplus A$ and $S_j=T \uplus B$. 
     Since $S_i$ is a vertex cover of $G$ that is disjoint from $B$, we can see that $B$ is an independent set. Similarly, $A$ is also an independent set. 
     Hence, $H=G[A \uplus B]$ is a bipartite graph.  Further, since $|S_i|=|S_j|$ we also have $|A|=|B|$. 
     \begin{clm}\label{H_pm}
      $H$ has a perfect matching.
     \end{clm}
      \begin{proof}[of Claim~\ref{H_pm}]
      Note that $U \subseteq T$. Consider any  $B' \subseteq B$. Since $S_i= T \uplus A$ is a vertex cover of $G$, we have $N_G(B') \subseteq T \uplus A$.
If $|N_H(B')| < |B'|$, then $S'=T \uplus (B \setminus B') \uplus N_H(B')$
is a vertex cover of size smaller than $k$ with $U \subseteq S'$, violating the fact that $\mvc_{_{U}}(G)=k$. 
Therefore, $\forall B' \subseteq B$, $|N_H(B')| \ge |B'|$ and by Hall's theorem\cite{diestel}, $H$ has a perfect matching.
\qed
      \end{proof}
      Since $v \in S_j \setminus S_i$, we have $|A|=|B|\ge 1$.
       \begin{clm}\label{Hxv_pm}
   $\forall x \in A$, the bipartite graph $H \setminus \{x,v\}$ has a perfect matching.
     \end{clm}
      \begin{proof}[of Claim~\ref{Hxv_pm}]
      If $H \setminus \{x,v\}$ is empty, then the claim holds trivially. 
      Consider any non-empty subset $B' \subseteq (B \setminus \{v\})$. By Claim~\ref{H_pm}, $|N_H(B')| \ge |B'|$. 
      If $|N_H(B')| = |B'|$, then
$S'=T \uplus (B \setminus B') \uplus N_H(B')$ is a vertex cover of $G$ with $|S'|=k$ and $(U \cup \{v\}) \subseteq S'$. 
This contradicts the choice of $S_j$, since the symmetric difference of $S'$ and $S_i$ has lesser cardinality
than that of $S_i$ and $S_j$. Therefore, $|N_H(B')| \ge |{B'}|+1$ and $|N_H(B') \setminus \{x\}| \ge |{B'}|$.
Hence, for all subsets $B' \subseteq (B \setminus \{v\})$, $|N_H(B') \setminus \{x\}| \ge |{B'}|$ and by Hall's theorem, $H \setminus \{x,v\}$ has a perfect matching.
\qed
      \end{proof}
With the help of Claim~\ref{H_pm} and Claim~\ref{Hxv_pm}, we can now complete the proof of Lemma~\ref{lem:evc-subset-converse}.       
We will describe how the attack on the edge $uv$ can be defended by moving guards.  
      \begin{itemize}
       \item Case 1. $u \in A$:\\
       By Claim~\ref{Hxv_pm}, there exists a perfect matching $M$ in $H \setminus \{u,v\}$. In order to defend 
       the attack, move the guard on $u$ to $v$ and also all the guards on $A \setminus u$ to $B \setminus v$ 
       along the edges of the matching $M$. 
       \item Case 2. $u \in T$:\\
       Recall that $|A|=|B|\ge 1$. 
       By our assumption, the vertex cover $S_i=T \uplus A$ is connected. Let $P$ be a shortest path from $A$ to $u$ in $G[A \cup T]$.
       By the minimality of $P$, it has exactly one vertex $x$ from $A$ and $x$ will be an endpoint of $P$. Suppose 
       $P=(x, z_1, z_2, \cdots, z_t=u)$ where $z_i \in T$, for $1 \le i \le t$. By Claim~\ref{Hxv_pm}, 
       there exists a perfect matching $M$ in $H \setminus \{x,v\}$. In order to defend 
       the attack, move the guard on $u$ to $v$, $x$ to $z_1$ and $z_i$ to $z_{i+1}$, $\forall i\in [t-1]$.
       In addition, move all the guards on $A \setminus \{x\}$ to $B \setminus \{v\}$ 
       along the edges of the matching $M$. 
      \end{itemize}
In both cases, the attack can be defended by moving the guards as mentioned and the new configuration is $S_j$. 
\hfill\qed\end{proof}     
% % % 
% % % 
% % %  If $\evc(G) = \mvc(G)$, then for every vertex $v \in V \setminus X$, there is a minimum vertex cover $S_v$ of $G$ such that $X \cup \{v\} \subseteq S_v$. 
The following theorem, which follows from Theorem~\ref{thm:necessary-condn} and Lemma~\ref{lem:evc-subset-converse}, gives a necessary 
and sufficient condition for a graph $G$ to satisfy $\evc(G)=\mvc(G)$, if
every minimum vertex cover of $G$ that contains all cut vertices is connected.
\begin{theorem}[Characterization Theorem]\label{thm:evc-cut-vertices-char}
 Let $G(V, E)$ be a graph that belongs to $\mathcal{F}$, with at least two vertices, and $X \subseteq V$ be the set of cut vertices of $G$. 
 Then, $\evc(G)=\mvc(G)$ if and only if for every vertex $v \in V \setminus X$, there exists a minimum vertex cover $S_v$ of $G$ such 
 that $(X\cup \{v\}) \subseteq S_v$. 
\end{theorem}
\begin{proof}
Let $k=\mvc(G)$ and suppose every minimum vertex cover $S$ of $G$ with $X \subseteq S$ is connected. 

If for every vertex $v \in V \setminus X$ there exists a minimum vertex cover $S_v$ of $G$ such that $(X\cup \{v\}) \subseteq S_v$, 
then it is easy to  see that $\mvc_{_{X}}(G)=\mvc(G)$. Hence, by our assumption that every minimum vertex cover $S$ of $G$ with $X \subseteq S$ is connected,
it follows that every vertex cover of $G$ of cardinality $\mvc_{_X}(G)$ that contains $X$ is connected.
Therefore, by Lemma~\ref{lem:evc-subset-converse}, we have $\evc_{_{X}}(G)=\mvc_{_{X}}(G)=\mvc(G)$.
Since $\mvc(G) \le \evc(G) \le \evc_{_{X}}(G)$, it follows that $\evc(G)=\mvc(G)$. 

Conversely, if $\evc(G)=\mvc(G)$, by Theorem~\ref{thm:necessary-condn}, 
for every vertex $v \in V \setminus X$, there exists a minimum vertex cover $S_v$ of $G$ such that $(X\cup \{v\}) \subseteq S_v$.
\hfill\qed\end{proof}
\begin{remark}
 By going through the proofs presented, it can be verified that Theorem~\ref{thm:necessary-condn} is valid also for the variant of the game where more than one guard is allowed
on a vertex simultaneously. 
\end{remark}
In the next section, we discuss some algorithmic implications of this theorem.
\section{Algorithmic consequences of the characterization theorem}
In this section, we derive some computational upper bounds that can be derived using the characterization theorem. 

The corollary below gives a method to determine $\evc(G)$ for a connected graph $G$ such that all its minimum vertex covers are connected.
\begin{corollary}\label{cor:evc-computation}
Let $G(V, E)$ be a connected graph for which every minimum vertex cover is connected. 
If for every vertex $v \in V$, there exists a minimum vertex cover $S_v$ of $G$ such that $v \in S_v$, then $\evc(G)=\mvc(G)$. Otherwise, $\evc(G)=\mvc(G)+1$.
\end{corollary}
\begin{proof}
Klostermeyer et~al.~\cite{Klostermeyer2009} showed that $\evc(G)$ is at most one more than the size of a connected vertex cover of $G$. 
Hence, from our assumption that all minimum vertex covers of $G$ are connected, we have $\evc(G) \le \mvc(G)+1$. 
Now, the result follows by Theorem~\ref{thm:evc-cut-vertices-char}.
\hfill\qed\end{proof}
\begin{remark}\label{rmk:evc-np}
If $G(V, E)$ is a connected graph for which every minimum vertex cover is connected, then it is easy to see that\\ 
$\evc(G) = \min \{k : \forall v \in V(G), G\text{ has a vertex cover of size $k$ containing $v$}\}$.\\
This is because for any vertex $v$, there is a vertex cover of $G$ of cardinality $\mvc(G)+1$ that contains $v$.
\end{remark}
\begin{corollary}\label{cor:evc-bic-np}
 Given a connected graph $G$ for which every minimum vertex cover is connected, deciding whether $\evc(G) \le k$ is in NP.
\end{corollary}
\begin{proof}
 By Remark~\ref{rmk:evc-np}, it is easy to get a polynomial time verifiable certificate to check if $\evc(G) \le k$. 
 The certificate can consist of $n$ vertex covers, in which for each vertex $v$ of $G$, there is a vertex cover $S_v$ of size $k$ that contains $v$.  
\hfill\qed
\end{proof}
The following is another immediate corollary of Theorem~\ref{thm:evc-cut-vertices-char}.
\begin{corollary}\label{cor:evc-sigma2}
Given a graph $G$ that belongs to $\mathcal{F}$, deciding whether $\evc(G)=\mvc(G)$ is in $P^{NP}$. 
\end{corollary}
\begin{proof}
Using polynomially many queries to an NP oracle, we can compute $\mvc(G)$. 
 Let $t=\mvc(G)$ and $X$ be the set of cut vertices of $G$. Computing $X$ can be done in polynomial time. 
 By Theorem~\ref{thm:evc-cut-vertices-char}, it suffices to check whether for every vertex $v \in V \setminus X$, 
 there exists a vertex cover $S_v$ of $G$ of size $t$ such that $(X\cup \{v\}) \subseteq S_v$. 
 Checking whether there exists a vertex cover $S_v$ of $G$ of size $t$ such that $(X\cup \{v\}) \subseteq S_v$ is equivalent to checking 
 whether the graph $G\setminus (X \cup \{v\})$ has a vertex cover of size $t-|X|-1$. This decision problem is also in NP.  
 Thus, the entire procedure of deciding whether $\evc(G)=\mvc(G)$ requires only polynomially many queries to an NP oracle. 
 \qed
\end{proof}
Now, let us look at some graph classes for which the results stated above are applicable. 
\subsection{Locally connected graphs}
A graph $G$ is \textit{locally connected} if for every vertex $v$ of $G$, its open neighborhood $N_G(v)$ induces a connected subgraph in $G$.
Erd\"{o}s, Palmer and Robinson \cite{erdos1983} showed that local connectivity of random graphs exhibits a sharp threshold phenomenon. 
They proved that, when probability of adding an edge, $p(n)$, is $\sqrt{\frac{3 \log n}{2 n}}$ or higher, almost all graphs in $\mathcal{G}{(n, p)}$ 
are locally connected. 
Some other sufficient conditions for a graph to be locally connected were given by Chartrand and Pippert \cite{chartrand1974} and 
Vanderjagt \cite{vanderjagt1974}.\\\\
\noindent A \textit{block} in a connected graph $G$ is either a maximal biconnected component or a bridge of $G$. 
The following is a property of graphs for which each block is locally connected.
\begin{property}\label{prop:locally-connected}
 Let $G(V, E)$ be a connected graph. If every block of $G$ is locally connected, then every vertex cover of $G$ that contains all its cut vertices is 
 a connected vertex cover.
 \end{property}
\begin{proof}
 The restriction of a vertex cover $S$ of $G$ to a block will give a vertex cover of the block. 
 Hence, to prove the observation, it is enough to show that all vertex covers of a locally connected graph $G$ are connected.
 
 For contradiction, suppose $G$ is a locally connected graph and $S$ is a vertex cover of $G$ such that $G[S]$ is not connected.
 Then, there exists a vertex $v \in V\setminus S$ and two components $C_1$ and $C_2$ of $G[S]$ such that $v$ is adjacent to vertices $v_1 \in V(C_1)$ 
 and $v_2 \in V(C_2)$.
 Since $S$ is a vertex cover that does not contain $v$, we have $N_G(v) \subseteq S$. Since $G$ is locally connected, we know that $N_G(v)$ is connected and
 therefore, $v_1$ and $v_2$ must belong to the same component of $G[S]$, which is a contradiction.
 Hence, $G[S]$ is connected.
 \hfill\qed\end{proof}
 From Property~\ref{prop:locally-connected}, it follows that $\mathcal{F}$ includes all graphs for which every block is locally connected and 
 therefore, the conclusion in Theorem~\ref{thm:evc-cut-vertices-char} applies for them. 
 Combining this with Corollary~\ref{cor:evc-computation}, Corollary~\ref{cor:evc-bic-np} and Corollary~\ref{cor:evc-sigma2},
 we get:
\begin{corollary}\label{cor:local}
For a locally connected graph $G$, $\evc(G) \in \{\mvc(G), \mvc(G)+1\}$ and deciding whether $\evc(G) \le k$ is in NP.\\ 
If $G$ is a connected graph with at least two vertices in which every block is locally connected, then
  \begin{itemize}
   \item $\evc(G)=\mvc(G)$ if and only if for every vertex $v$ of $G$ that is not a cut-vertex, there is a minimum vertex cover of 
  $G$ that contains $v$ and all the cut-vertices
  \item deciding whether $\evc(G)=\mvc(G)$ is in $P^{NP}$.
  \end{itemize}
\end{corollary}
The hardness of computing eternal vertex cover number of locally connected graphs is discussed in Section~\ref{sec:complexity}.
\subsubsection{Internally triangulated planar graphs}
A graph is an \textit{internally triangulated planar graph} if it has a planar embedding in which all internal faces are triangles. 
It can be easily seen that biconnected internally triangulated planar graphs are locally connected.  
Hence, both the conclusions of Corollary~\ref{cor:local} are applicable to internally triangulated planar graphs, as stated in Section~\ref{sec:intro}.  
The complexity of computing eternal vertex cover number of locally connected graphs is discussed in Section~\ref{sec:complexity}.

Since biconnected chordal graphs are locally connected, conclusions of Corollary~\ref{cor:local} hold for chordal graphs as well. However, for
chordal graphs, we can derive some stronger results, as explained below. 
\subsection{Polynomial time algorithms} 
A class of graphs $\mathcal{H}$ is called hereditary, if deletion of a subset of vertices from any graph $G$ in $\mathcal{H}$ would always yield
another graph in $\mathcal{H}$. 

We show that for any hereditary subclass of $\mathcal{F}$ for which minimum vertex cover computation is polynomial time, some stronger algorithmic
consequences follow. 
\begin{corollary}\label{cor:evc-hereditary}
Let $\mathcal{H}$ be a hereditary subclass of $\mathcal{F}$ such that, for all graphs $G$ in $\mathcal{H}$, $\mvc(G)$ can be computed in polynomial time.
Then, given a graph $G$ that belongs to $\mathcal{H}$,
\begin{enumerate}
 \item deciding whether $\evc(G)=\mvc(G)$ can be done in polynomial time,
 \item if $\evc(G)=\mvc(G)$, then there is a polynomial time (per-round) strategy for
 guard movements using $\evc(G)$ guards, and 
\item if $G$ is biconnected, then $\evc(G)$ can be computed in polynomial time and 
  there is a polynomial time (per-round) strategy for guard movements using $\evc(G)$ guards.
\end{enumerate}
\end{corollary}
\begin{proof}
Without loss of generality, we may assume that $G$ is connected and has at least two vertices. 
 \begin{enumerate}
  \item By our assumption, we can compute $\mvc(G)$ in polynomial time. 
  Identifying the set of cut vertices $X$ of $G$ can also be done in polynomial time. 
  By Theorem~\ref{thm:evc-cut-vertices-char}, to decide whether $\mvc(G)=\evc(G)$, it is enough to check for every vertex $v \in V \setminus X$ 
  whether $G$ has a minimum vertex cover $S_v \supseteq X \cup \{v\}$. 
  Note that checking whether $G$ has a minimum vertex cover containing $X \cup \{v\}$ 
  is equivalent to checking whether $\mvc(G)=\mvc(G')+|X|+1$, 
  where $G' = G \setminus (X \cup \{v\})$. Since $G' \in \mathcal{H}$, 
  we can compute $\mvc(G')$ and perform this checking in polynomial time.  
  \item Suppose $\evc(G)=\mvc(G)=k$ and $X$ be the set of cut vertices of $G$. 
  By Lemma~\ref{lem:evc-cut-vertices}, $\evc_{_{X}}(G)=\mvc_{_{X}}(G)=k$.  
  By our assumption, every vertex cover $S$ of $G$ with $X \subseteq S$ is a connected vertex cover. 
  Therefore, by Observation~\ref{obs:evc-subset}, for every vertex $v \in V \setminus X$, $\mvc_{_{X \cup \{v\}}}(G)=k$. 
  
  We complete the proof by extending the basic ideas used in the proof of Lemma~\ref{lem:evc-subset-converse}.
  Take any minimum vertex cover $S$ of $G$ with $X \subseteq S$ as the starting configuration. 
  It is enough to show that from any minimum vertex cover $S_i$ of $G$ with $X \subseteq S_i$, following an attack on an edge $uv$ 
  such that $u \in S_i$ and $v \notin S_i$, we can safely defend the attack by moving to a minimum vertex cover $S_j$ 
  such that $(X \cup \{v\}) \subseteq S_j$. Consider an attack on such an edge $uv$. 
  To start with, choose an arbitrary minimum vertex cover $S'$ of $G$ with $(X \cup \{v\}) \subseteq S'$ as a candidate for being the next configuration. 
  Suppose $Z=S_i \cap S'$, $S_i= Z \uplus A$ and $S'=Z \uplus B$. By similar arguments as in the proof of Lemma~\ref{lem:evc-subset-converse}, 
  $H=G[A \uplus B]$ is a non-empty bipartite graph with a perfect matching. 
  \begin{enumerate}
   \item[i.]
  If for each $x \in A$, the bipartite graph $H \setminus \{x,v\}$ has a perfect matching, then we can choose $S'$ to be the new configuration $S_j$ and move guards as explained in 
  the proof of Lemma~\ref{lem:evc-subset-converse}. 
  \item[ii.] If the bipartite graph $H \setminus \{x,v\}$ does not have a perfect matching for some $x \in A$, there is a method to redefine $S'$ to get 
  a new candidate configuration, as described below. 
  
  In polynomial time
  we can identify a subset $B' \subseteq (B \setminus \{v\})$ for which $|B'| > |N_{H\setminus \{x, v\} }(B')|$, using a standard procedure described below.
  First find a max-matching $M$ in $H \setminus \{x, v\}$ and identify an unmatched vertex $y \in B \setminus \{v\}$. Let $B'$ be the set of vertices in $B \setminus \{v\}$ 
  reachable via $M$-alternating paths from $y$ in $H \setminus \{x,v\}$, together with vertex $y$. If $|B'| \le |N_{H \setminus \{x,v\}}(B')|$, 
  it would result in an $M$-augmenting path
  from $y$, contradicting the maximality of $M$. Thus, $|B'| > |N_{H\setminus \{x, v\} }(B')|$.  
  (In fact, since $H$ has a perfect matching, $|B'| \le |N_{H}(B')|$ and this would mean $x \in N_{H}(B')$.)  
  Now, let $S'' = Z \cup \{x\} \cup (B \setminus B') \cup N_{H\setminus \{x, v\}}(B')$. 
  It is easy to see that $S''$ is a minimum vertex cover of $G$ with $(A \cup \{v\}) \subseteq S''$ and the symmetric difference of $S''$ and $S_i$ 
  is smaller than the symmetric difference of $S'$ and $S_i$. 
  Now we redefine $S'$ to be $S''$ and iterate the steps above after redefining the sets $Z$, $B$ and $A$ and the graph $H$ according to the new $S'$.
 \end{enumerate}
  We will repeat these steps until we reach a point when  
  the (re-defined) bipartite graph $H \setminus \{x,v\}$ has a perfect matching, for each $x \in A$. This process will terminate in less than $n$ iterations, because in each iteration, 
  the symmetric difference of the candidate configuration with $S_i$ is decreasing. 
  The basic computational steps involved in this process are computing minimum vertex covers containing $X \cup \{v\}$, finding maximum matching in
  some bipartite graphs and computing some alternating paths.
  All these computations can be performed in polynomial time \cite{hopcroft1973}. 
  \item Let $G$ be a biconnected graph in $\mathcal{H}$. By Corollary~\ref{cor:evc-computation}, $\evc(G) \in \{ \mvc(G)$, $\mvc(G)+1\}$.  Therefore, by using part 1 of this corollary,
  $\evc(G)$ can be decided exactly, in polynomial time. If  $\evc(G) = \mvc(G)$, using part 2 of this theorem, we can complete the proof. 
  If $\evc(G) = \mvc(G)+1$, we will make use of the fact that every minimum vertex cover of $G$ is connected. We will fix a minimum vertex cover $S$
  and initially place guards on all vertices of $S$ and also on one additional vertex. 
  Using the method given by Klostermeyer et al.~\cite{Klostermeyer2009} to show that $\evc(G) \le \cvc(G)+1$, we will be able 
  to keep defending attacks while maintaining guards 
  on all vertices of $S$ after end of each round of the game.
  \end{enumerate}
\hfill\qed\end{proof}
\subsubsection{Chordal graphs}
A graph is \textit{chordal} if it contains no induced cycle of length four or more. 
It is well-known that chordal graphs form a hereditary graph class and computation of a minimum vertex 
cover of a chordal graph can be done in polynomial time \cite{Gavril1972}. 
It can also be easily seen that biconnected chordal graphs are locally connected.  
Hence, the conclusions of Corollary~\ref{cor:evc-hereditary} hold for chordal graph.
\section{Complexity results}\label{sec:complexity}
In this section, we discuss some computational lower bounds of the eternal vertex cover problem. 
Fomin et al.~\cite{Fomin2010} showed that, given a graph $G$ and an integer $k$, deciding whether $\evc(G) \le k$ is NP-hard. 
However, the graph obtained by their reduction is not locally connected and it seems to be unknown whether the problem is NP-complete 
for any graph classes. In general, it is not known whether this problem is in NP or not. 
We show that this problem is NP-complete for locally connected graphs and 
biconnected internally triangulated planar graphs. Approximation hardness of the problem for locally connected graphs is also studied here. 
\subsection{Eternal vertex cover number of locally connected graphs}\label{sec:comp-local}
\begin{proposition}\label{prop:local_evc_NPC}
 Given a locally connected graph $G$ and an integer $k$, it is NP-complete to decide if $\evc(G) \le k$. Moreover, it is NP-hard to approximate
$\evc(G)$ of locally connected graphs within any factor smaller than $10 \sqrt{5}-21$ unless P=NP.
\end{proposition}
\begin{proof}
By Corollary \ref{cor:local}, given a locally connected graph $G$, and an integer $k$, deciding whether $\evc(G) \le k$ is in NP.

 A famous result by Dinur et al.~\cite{dinur2005} states that given a connected graph $G$, it is NP-hard to approximate
the minimum vertex cover number of connected graphs within any factor smaller than $10 \sqrt{5}-21$. 
For a given connected graph $G$ and integer $k$, 
we can construct a locally connected graph $G'$ by adding a new vertex to $G$ and connecting it to all the existing vertices of $G$. 
 It can be seen easily that $\mvc(G')=\mvc(G)+1$. 
Therefore, even for locally connected graphs, the minimum vertex cover number is NP-hard to approximate within any factor smaller than $10 \sqrt{5}-21$. 
By Corollary \ref{cor:local}, $\mvc(G') \le \evc(G') \le \mvc(G')+1$. Hence, the result follows.
\hfill\qed\end{proof}
\subsection{Eternal vertex cover number of biconnected internally triangulated planar graphs}\label{sec:comp-int-triang}
Since biconnected internally triangulated graphs are locally connected, as explained in the previous section, 
given a biconnected internally triangulated planar graph $G$ and an integer $k$, deciding whether $\evc(G) \le k$ is in NP. 
We will show that this decision problem is NP-hard using a sequence of simple reductions. 
First we show that the classical vertex cover 
problem is NP-hard for biconnected internally triangulated planar graphs. Then we will show that an additive one approximation to vertex cover is also NP-hard 
for the same class and use it to derive the required conclusion. 
\begin{figure}[h]
\centering 
\includegraphics[scale=.6]{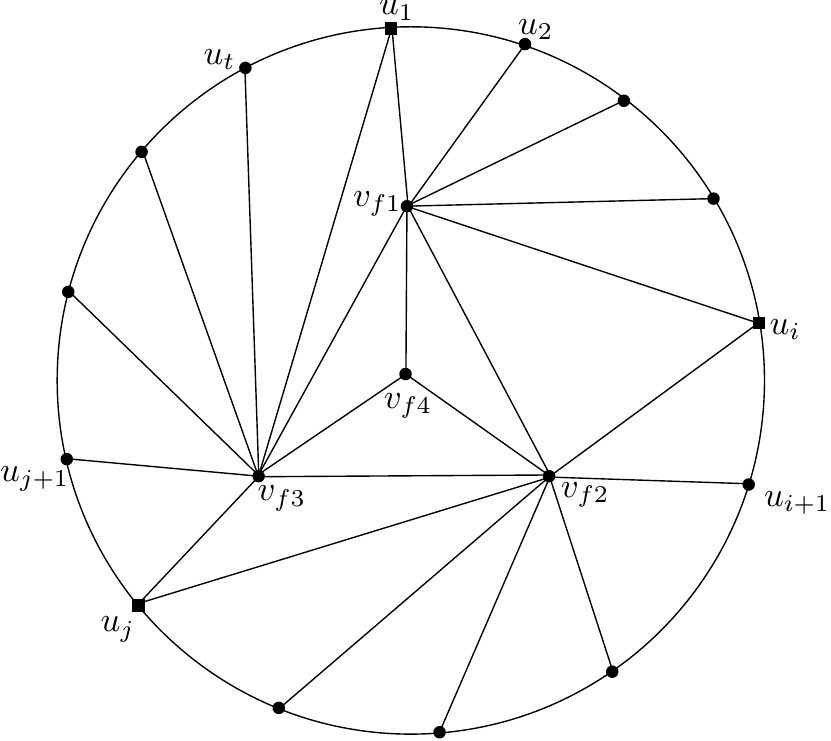}
\caption{Triangulating an internal face $f=u_1,u_2,\cdots u_t,u_1$ of $G$ with $t>3$, by adding new vertices and edges.}
\label{fig:fig_hardness}
 \end{figure}
\begin{proposition}\label{prop:triang_vc_NPC}
 Given a biconnected internally triangulated planar graph $G$ and an integer $k$, it is NP-complete to decide if $\mvc(G) \le k$.
\end{proposition}
\begin{proof}
  The vertex cover problem on biconnected planar graph is known to be NP-hard \cite{MOHAR2001102}. 
 We show a reduction from the 
 vertex cover problem on biconnected planar graph to the 
 vertex cover problem on biconnected internally triangulated planar graph. 
 Suppose we are given a biconnected planar graph $G$ and an integer $k$. We construct $G'$ such that $G$ is an induced subgraph in $G'$. 
 First, compute a planar 
 embedding of $G$ in polynomial time \cite{schnyder1990}. We know that, in any planar embedding, each face of a biconnected planar graph is bounded by a cycle \cite{diestel2018}. 
 To construct $G'$, each internal face of $G$ with more than three vertices on its boundary is triangulated by adding four new vertices and some edges (see Fig.~\ref{fig:fig_hardness}).
 Let $f=u_1,u_2,\cdots u_t,u_1$ be a cycle bounding an internal face of $G$, with $t>3$. Let
 $i$ and $j$ be two distinct indices from $[2,t]$. Add three vertices $v_{f1}$, $v_{f2}$ and $v_{f3}$ inside $f$. Now, add edges ($v_{f1}$, $u_1$), ($v_{f1}$, $u_2$) $\cdots$
 ($v_{f1}$, $u_i$), ($v_{f2}$, $u_i$), ($v_{f2}$, $u_{i+1}$) $\cdots$ ($v_{f2}$, $u_j$), ($v_{f3}$, $u_j$), ($v_{f3}$, $u_{j+1}$), $\cdots$ ($v_{f3}$, $u_t$) 
 and ($v_{f3}$, $u_1$) 
 in such a way that the graph being constructed does not loses its planarity. Add a new vertex $v_{f4}$ inside the triangle formed by $v_{f1}$, $v_{f2}$ 
 and $v_{f3}$. Now, make $v_{f4}$ adjacent to $v_{f1}$, $v_{f2}$ and $v_{f3}$ by adding edges ($v_{f4},v_{f1}$), ($v_{f4},v_{f2}$) and ($v_{f4},v_{f3}$). 
 Repeat this construction procedure for all faces of $G$ bounded by more than $3$ vertices. 
 As per the construction, it is clear that the resultant graph $G'$ is biconnected, internally triangulated and planar.
 It can be seen easily that the biconnected triangulated planar graph $G$ has a vertex cover of size at most $k$ if and only if the biconnected 
 internally triangulated planar graph $G'$ has a vertex 
 cover of size at most $k'=k+3f'$ where $f'$ is the number of internal faces of $G$ bounded by more than $3$ vertices. 
\hfill\qed\end{proof}

 In the proof of Proposition~\ref{prop:local_evc_NPC}, we used the  APX-hardness of vertex cover problem of locally connected graphs to 
 derive the APX-hardness of eternal vertex cover problem of locally connected graphs.
 However, a polynomial time approximation scheme is known for computing the minimum vertex cover number of planar graphs \cite{Baker1994}. 
 Hence, we need a different approach to show the NP-hardness of eternal vertex cover problem of planar graphs.
  We will show that if minimum vertex cover number of biconnected internally triangulated planar graphs can be approximated within an additive one error, 
  then it can be used to precisely compute the minimum vertex cover number of graphs of the same class. 
\begin{figure}[h]
\centering 
\includegraphics[scale=.7]{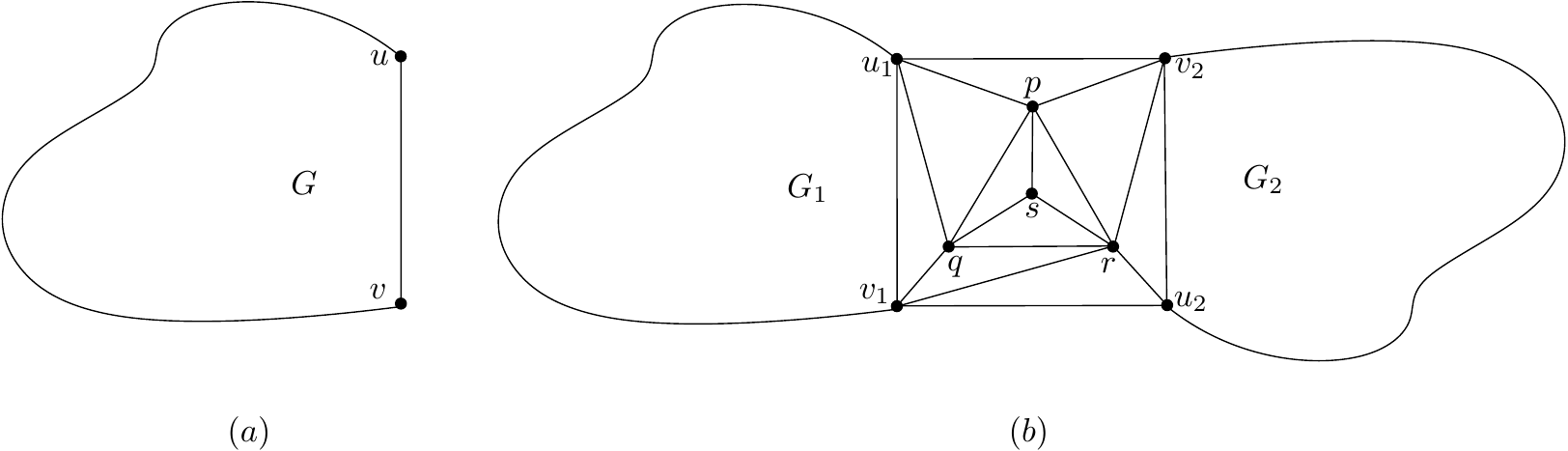}
\caption{NP-hardness reduction for additive one approximation of vertex cover number of biconnected internally triangulated planar graphs.}
\label{fig:fig_approx_hardness}
 \end{figure}
  \begin{proposition}\label{prop:additive_vc_NPC}
 Getting an additive $1$-approximation for computing the minimum vertex cover number of biconnected internally triangulated planar graphs is NP-hard. 
\end{proposition} 
\begin{proof}
Let $G(V, E)$ be the given biconnected internally triangulated planar graph. Consider a fixed planar internally triangulated embedding of $G$. 
The reduction algorithm constructs a new graph $G'$ as follows. Make two copies of $G$ 
 namely, $G_1$ and $G_2$. For each vertex $v \in V$, let $v_1$ and $v_2$ denote its corresponding vertices in $G_1$ and $G_2$ respectively. 
 Choose any arbitrary edge $e=uv$ on the outer face of $G$. 
 Add new edges $u_1v_2$ and $v_1u_2$ maintaining the planarity. Now, the new graph is biconnected and planar; but the face with boundary $u_1 v_1 u_2 v_2$
 needs to be triangulated. For this, we follow the same procedure we used in the proof of Proposition~\ref{prop:triang_vc_NPC} which adds four new 
 vertices $p, q, r, s$ and some new edges inside this face (see Fig.~\ref{fig:fig_approx_hardness}). The resultant graph $G'$ is biconnected, internally triangulated and planar. 
 
Consider a minimum vertex cover $S$ of $G$ such that $\mvc(G)=|S|=k$. It is clear that either $u$ or $v$ is in $S$. It is easy to see that
$S'=\{v_1 : v\in S\} \cup \{v_2 : v \in S\} \cup \{ p, q, r\}$ is a vertex cover of $G'$ with size $2k+3$. Similarly, at least $k$ vertices from 
$G_1$ and $G_2$ and at least $3$ vertices among \{$p$, $q$, $r$, $s$\} has to be chosen for a minimum vertex cover of $G'$. This shows that
$\mvc(G')=2k+3$.
 
  Suppose there exist a polynomial time additive $1$-approximation algorithm for computing the minimum vertex cover number of biconnected internally triangulated graphs. 
  Let $k'$ be the approximate value of minimum vertex cover of $G'$, computed by this algorithm. 
  Then, $k' \in \{\mvc(G'),\mvc(G')+1\}$. This implies that $\mvc(G)=\left\lfloor{\frac{(k'-3)}{2}}\right\rfloor$, giving a polynomial time algorithm to compute $\mvc(G)$. 
  Hence, by Proposition \ref{prop:triang_vc_NPC}, getting an additive $1$-approximation for computing the minimum vertex cover for biconnected
 internally triangulated planar graphs is NP-hard.
\hfill\qed\end{proof}
By Corollary~\ref{cor:local}, for a biconnected internally triangulated graph, $\mvc(G) \le \evc(G) \le \mvc(G)+1$. 
Therefore, a polynomial time algorithm to compute $\evc(G)$ would give a polynomial time additive $1$-approximation for $\mvc(G)$. 
Hence, by Proposition~\ref{prop:additive_vc_NPC}, we have the following result. 
\begin{proposition}
 Given a biconnected internally triangulated planar graph $G$ and an integer $k$, it is NP-complete to decide if $\evc(G)\le k$.
\end{proposition}
Note that, using the PTAS designed by Baker et al.~\cite{Baker1994} for computing the minimum vertex cover number of planar graphs, 
it is possible to derive a polynomial time approximation scheme for computing the eternal vertex cover number of biconnected internally triangulated planar graphs.
A summary of the complexity results presented here are given in Fig.~\ref{fig:complexity}. 

\begin{figure}[h]
\centering
\includegraphics[scale=0.7]{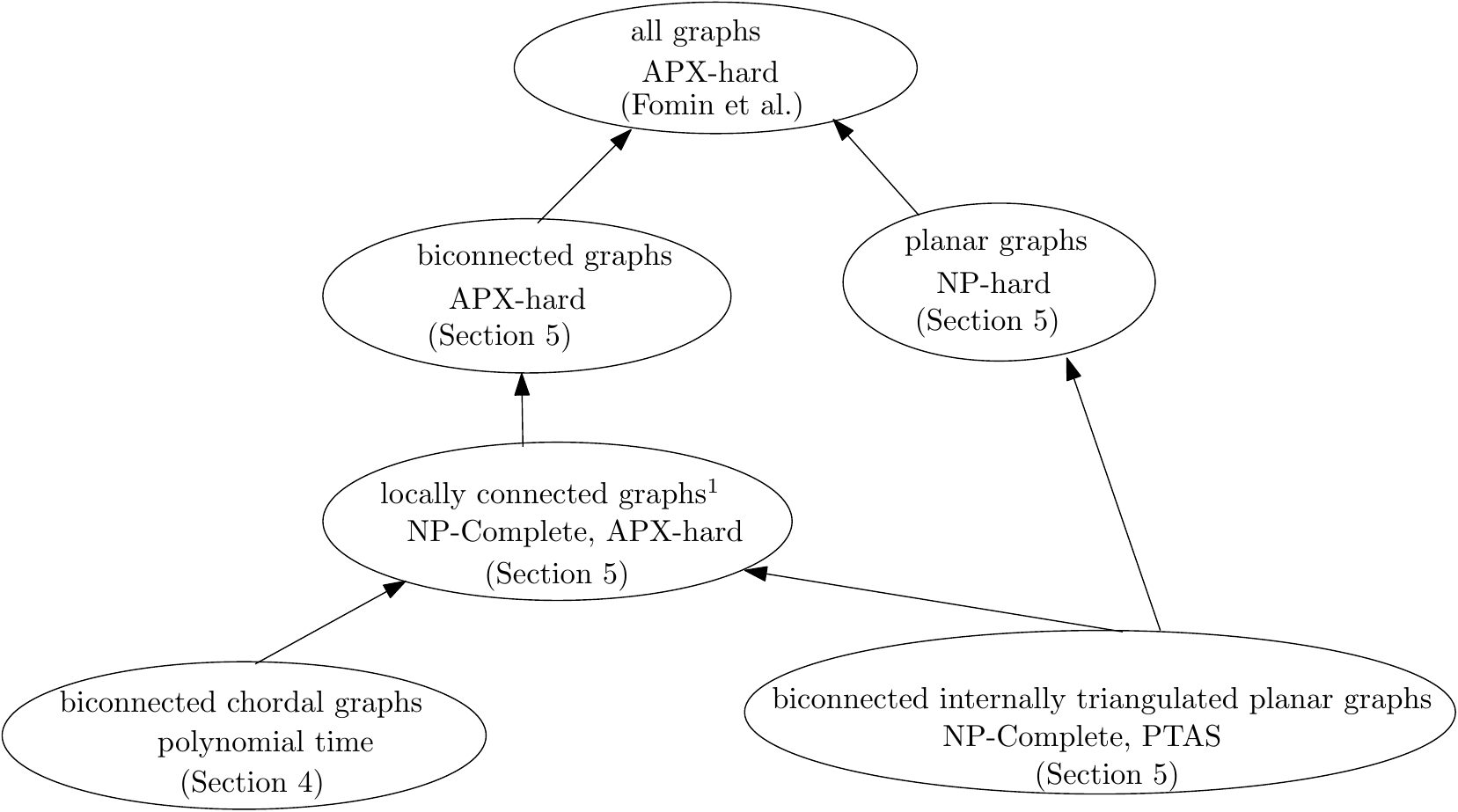}
\caption{Complexity of deciding whether $\evc(G)\le k$.}
\label{fig:complexity}
 \end{figure}
 \footnotetext[1]{All locally connected graphs are biconnected, with the exception of $K_2$.}
\section{Is the necessary condition sufficient?}
It is interesting to ask if the necessary condition stated in Theorem~\ref{thm:necessary-condn} is 
sufficient for all graphs. Here, we give a biconnected bipartite planar graph of maximum degree $4$ which answers this question in negative. 
Consider the bipartite graph 
$G(X \cup Y,E)$ with $X=\{x_1,x_2,x_3,x_4,x_5\}$ and $Y=\{y_1,y_2,y_3,y_4,y_5\}$ shown in Fig.~\ref{fig:fig_contradiction_bicon}. 
This graph consists of two copies of $K_{2,3}$
on vertex sets $\{x_1,x_2,x_3,y_4,y_5\}$ and $\{y_1,y_2,y_3,x_4,x_5\}$ connected by two edges $x_1 y_1$ and $x_4 y_4$.
From the figure, it can be easily seen that $\mvc(G)=5$ and it has only 
 one minimum vertex cover $S_{x_2}=\{x_1,x_2,x_3,x_4,x_5\}$ that contains $x_2$. 
 Therefore, for defending an attack on an edge incident on the vertex $x_2$, 
 the guards need to move to the configuration $S_{x_2}$. 
 In this configuration, when there is an attack on the edge $x_5y_5$, $G$ has to move to a configuration containing $y_5$. 
 The only minimum vertex covers of $G$ containing $y_5$ are 
 $S_1=\{y_1,y_2,y_3,y_4,y_5\}$, $S_2=\{x_4,x_5,y_4,y_5,x_1\}$ and $S_3=\{x_4,x_5,y_4,y_5,y_1\}$. 
 Since the edge $x_5y_5$ does not belong to any maximum matching of $G$, a transition from $S_{x_2}$ to
 $S_1$ is not legal. Configurations $S_2$ and $S_3$ both contain $x_5$. Following the attack on $x_5y_5$ in configuration $S_{x_2}$, 
 when the guard on $x_5$ moves to $y_5$, 
 no other guard can move to $x_5$, because no neighbor of $x_5$ is occupied in $S_{x_2}$. 
 Thus, transitions to $S_2$ and $S_3$ are also not legal. 
 Hence, the attack on $x_5y_5$ cannot be handled and therefore $\evc(G)\neq \mvc(G)$.  
 
\begin{figure}[h]
\centering 
\includegraphics[width=4.5cm,height=4.5cm]{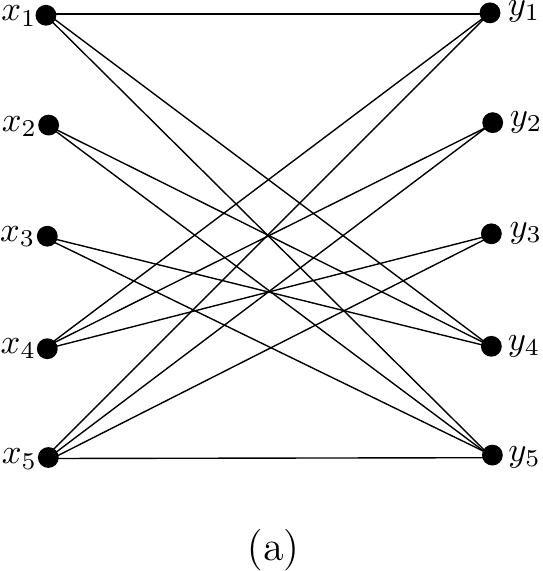}
\hspace*{.7cm}
\includegraphics[width=6cm,height=4.5cm]{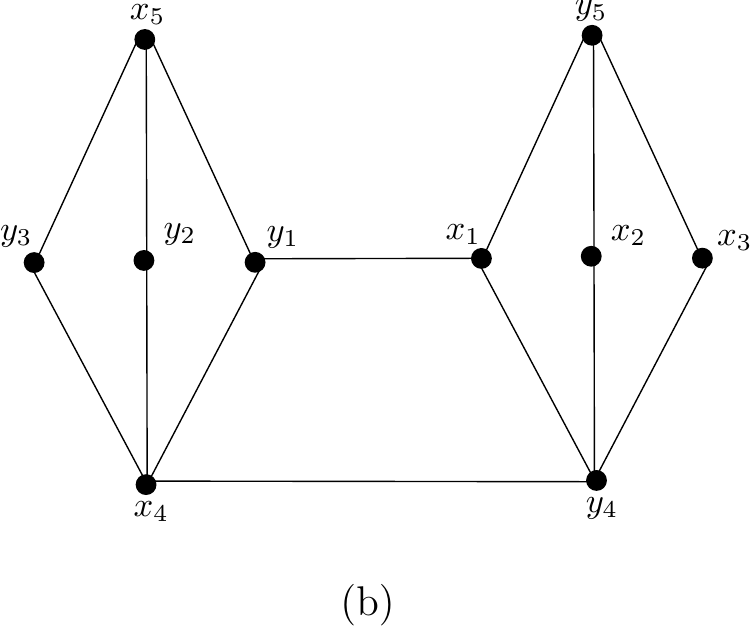}
\caption{(a) A biconnected bipartite planar graph with all vertices in some minimum vertex cover and $evc(G) \neq mvc(G)$ (b) A planar drawing of the same graph.}
\label{fig:fig_contradiction_bicon}
 \end{figure}
 This example shows that the necessary condition is not sufficient for planar graphs or bipartite graphs, even when they are biconnected. 
\section{A graph $G$ with an edge not contained in any maximum matching but $\mvc(G)=\evc(G)$}
\begin{figure}[h]
\centering
\includegraphics[scale=0.6]{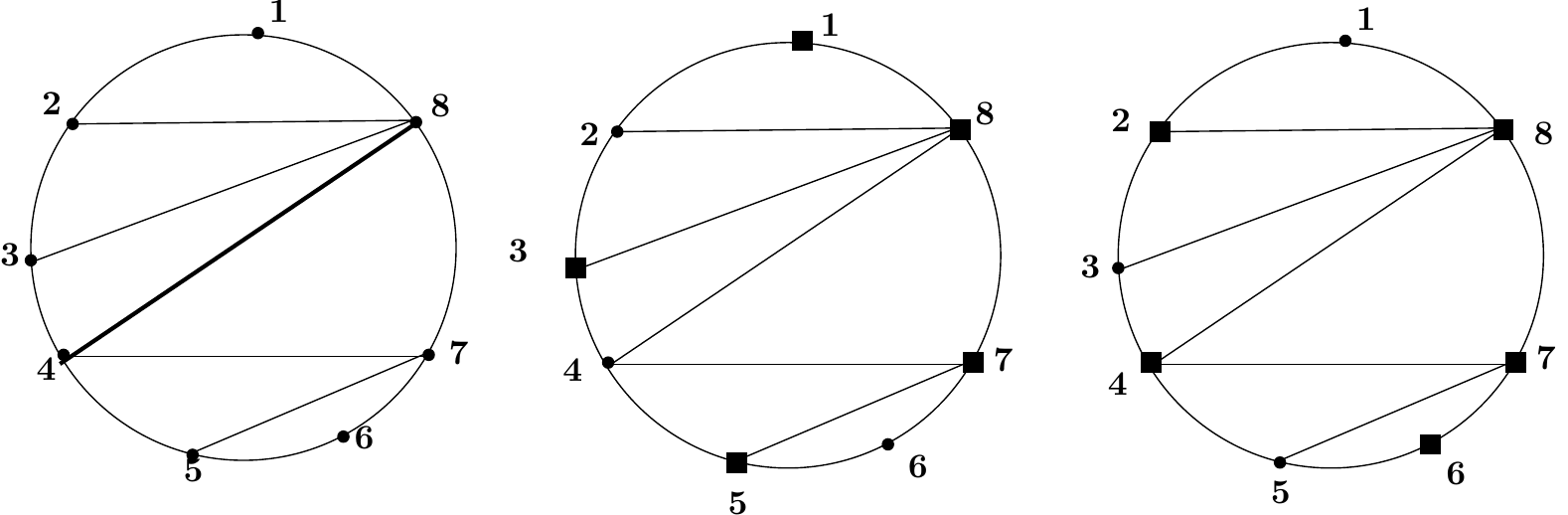}
\caption{$\mvc(G)=\evc(G)=5$ and size of maximum matching is 4. Edge $(8,4)$ is not contained in any maximum matching. }
\label{fig:fig_contradiction}
 \end{figure}
Klostermeyer et al.~\cite{Klostermeyer2009} proved that if a graph $G$ has two disjoint minimum vertex covers and each edge is contained in a maximum matching then 
$\mvc(G)=\evc(G)$. They had asked if $\mvc(G)=\evc(G)$, is it necessary that for every edge $e$ of $G$ there is a maximum matching of $G$ that contains $e$.
Here, we give a biconnected chordal graph $G$ for which the answer is negative.
The graph $G$ shown in Fig.~\ref{fig:fig_contradiction} has $\mvc(G)=5$, a maximum matching of size 4 and the edge $(8,4)$ not contained in any maximum matching. 
It can be shown that $\evc(G)=5$ because $G$ has an evc class with two configurations, $S_1=\{1,8,3,5,7\}$ and $S_2=\{6,8,2,4,7\}$.
 
Hence, even for a graph class $\mathcal{H}$ such that for all $G \in \mathcal{H}$, $\evc(G) \le \mvc(G)+1$,
there could be a graph $G \in \mathcal{H}$ with $\mvc(G)=\evc(G)$ and an edge not present in any 
maximum matching of $G$.
\section{Conclusion and open problems}
This paper presents an attempt to derive a characterization of graphs for which the eternal vertex cover number coincides with the vertex cover number.
A characterization that works for a graph class that includes chordal graphs and internally triangulated planar graphs is obtained. 
The characterization is derived from a simple to state necessary condition; and has several implications, including 
a polynomial time algorithm for deciding whether a chordal graph $G$ has $\evc(G)=\mvc(G)$ and a polynomial time algorithm for computing eternal vertex
cover number of biconnected chordal graphs. It would be interesting to study the complexity of eternal vertex cover problem of chordal graphs. 
A characterization that works for bipartite graphs also remains open. 

The characterization also leads to NP-completeness results for some graph
classes like locally connected graphs and biconnected internally triangulated planar graphs. 
Even though it was known that the general problem is NP-hard, 
to the best of our knowledge, results obtained here are the first NP-completeness results known for the eternal vertex cover problem.
\newpage
%\bibliography{evc-bib}

\begin{thebibliography}{10}

\bibitem{BabuCFPRW19}
Babu, J., Chandran, L.S., Francis, M.C., Prabhakaran, V., Rajendraprasad, D.,
  Warrier, J.N.:
\newblock On graphs with minimal eternal vertex cover number.
\newblock In: Algorithms and Discrete Applied Mathematics - 5th International
  Conference, {CALDAM} 2019, Kharagpur, India, February 14-16, 2019,
  Proceedings. (2019)  263--273

\bibitem{Klostermeyer2009}
Klostermeyer, W., Mynhardt, C.:
\newblock Edge protection in graphs.
\newblock Australasian Journal of Combinatorics \textbf{45} (2009)  235 -- 250

\bibitem{Fomin2010}
Fomin, F.V., Gaspers, S., Golovach, P.A., Kratsch, D., Saurabh, S.:
\newblock Parameterized algorithm for eternal vertex cover.
\newblock Information Processing Letters \textbf{110}(16) (2010)  702 -- 706

\bibitem{Anderson2014}
Anderson, M., Carrington, J.R., Brigham, R.C., D.Dutton, R., Vitray, R.P.:
\newblock Graphs simultaneously achieving three vertex cover numbers.
\newblock Journal of Combinatorial Mathematics and Combinatorial Computing
  \textbf{91} (2014)  275 -- 290

\bibitem{Klostermeyer2011}
Klostermeyer, W.F., Mynhardt, C.M.:
\newblock Graphs with equal eternal vertex cover and eternal domination
  numbers.
\newblock Discrete Mathematics \textbf{311} (2011)  1371 -- 1379

\bibitem{diestel}
Diestel, R.:
\newblock Graph Theory. Volume 173.
\newblock Springer, Graduate Texts in Mathematics (2000)

\bibitem{erdos1983}
Erd{\"o}s, P., Palmer, E.M., Robinson, R.W.:
\newblock Local connectivity of a random graph.
\newblock Journal of graph theory \textbf{7}(4) (1983)  411--417

\bibitem{chartrand1974}
Chartrand, G., Pippert, R.E.:
\newblock Locally connected graphs.
\newblock {\v{C}}asopis pro p{\v{e}}stov{\'a}n{\'\i} matematiky \textbf{99}(2)
  (1974)  158--163

\bibitem{vanderjagt1974}
Vanderjagt, D.W.:
\newblock Sufficient conditions for locally connected graphs.
\newblock {\v{C}}asopis pro p{\v{e}}stov{\'a}n{\'\i} matematiky \textbf{99}(4)
  (1974)  400--404

\bibitem{hopcroft1973}
Hopcroft, J.E., Karp, R.M.:
\newblock An $n^{5/2}$ algorithm for maximum matchings in bipartite graphs.
\newblock SIAM Journal on computing \textbf{2}(4) (1973)  225--231

\bibitem{Gavril1972}
Gavril, F.:
\newblock Algorithms for minimum coloring, maximum clique, minimum covering by
  cliques, and maximum independent set of a chordal graph.
\newblock SIAM Journal on Computing \textbf{1} (1972)  180 -- 187

\bibitem{dinur2005}
Dinur, I., Safra, S.:
\newblock On the hardness of approximating minimum vertex cover.
\newblock Annals of mathematics (2005)  439--485

\bibitem{MOHAR2001102}
Mohar, B.:
\newblock Face covers and the genus problem for apex graphs.
\newblock Journal of Combinatorial Theory, Series B \textbf{82}(1) (2001)
  102--117

\bibitem{schnyder1990}
Schnyder, W.:
\newblock Embedding planar graphs on the grid.
\newblock In: Proceedings of the first annual ACM-SIAM symposium on Discrete
  algorithms, Society for Industrial and Applied Mathematics (1990)  138--148

\bibitem{diestel2018}
Diestel, R.:
\newblock Graph theory.
\newblock Springer Publishing Company, Incorporated (2018)

\bibitem{Baker1994}
Baker, B.S.:
\newblock Approximation algorithms for \textsc{NP}-complete problems on planar
  graphs.
\newblock J. ACM \textbf{41}(1) (January 1994)  153--180

\end{thebibliography}

\end{document}